\documentclass{amsart}

\pdfoutput=1

\usepackage[utf8x,utf8]{inputenc} 
\usepackage[T1]{fontenc}
\usepackage{graphicx}

\usepackage[vlined,linesnumbered,boxruled,norelsize]{algorithm2e}
\IncMargin{0.25cm}
\DontPrintSemicolon{}

\renewcommand{\vec}[1]{\boldsymbol{#1}}

\DeclareMathOperator{\E}{\mathbb{E}}
\newcommand{\subsetsum}{\textsc{Subset Sum}}
\newcommand{\modsubsetsum}{\textsc{Modular Subset Sum}}
\newcommand{\Z}{\mathbb{Z}}

\newtheorem{theorem}{Theorem}[section]
\newtheorem{lemma}[theorem]{Lemma}
\newtheorem{definition}[theorem]{Definition}
\newtheorem{proposition}[theorem]{Proposition}

\makeatletter
\newtheorem*{rep@theorem}{\rep@title}
\newcommand{\newreptheorem}[2]{%
\newenvironment{rep#1}[1]{%
 \def\rep@title{#2 \ref{##1}}%
 \begin{rep@theorem}[restated]}%
 {\end{rep@theorem}}}
\makeatother

\newreptheorem{theorem}{Theorem}
\newreptheorem{lemma}{Lemma}
\newreptheorem{proposition}{Proposition}

\begin{document}

\title[Space--Time Tradeoffs for Subset Sum]{Space--Time Tradeoffs for Subset Sum:\\ An Improved Worst Case Algorithm}

\author{Per Austrin}
\address{Per Austrin, Aalto Science Institute, Aalto University, Finland and
KTH Royal Institute of Technology, Sweden}
\author{Petteri Kaski}
\address{Petteri Kaski, HIIT \&
Department of Information and Computer Science,
Aalto University, Finland}
\author{Mikko Koivisto}
\address{Mikko Koivisto, HIIT \&
Department of Computer Science, University of Helsinki, Finland}
\author{Jussi Määttä}
\address{Jussi Määttä, HIIT \&
Department of Information and Computer Science,
Aalto University, Finland}

\thanks{P.A.~supported by the Aalto Science Institute, the Swedish Research Council grant 621-2012-4546, and ERC Advanced Investigator grant 226203.  P.K.~supported by the Academy of Finland, grants 252083 and 256287.  M.K.~supported by the Academy of Finland, grants 125637, 218153, and 255675.}

\begin{abstract}

The technique of Schroeppel and Shamir (SICOMP, 1981) has long been
the most efficient way to trade space against time for the
\subsetsum{} problem. In the random-instance setting, however,
improved tradeoffs exist.  In particular, the recently discovered
dissection method of Dinur et al.\ (CRYPTO 2012) yields a
significantly improved space--time tradeoff curve for instances with
strong randomness properties.  Our main result is that these strong
randomness assumptions can be removed, obtaining the same space--time
tradeoffs in the worst case.  We also show that for small space usage
the dissection algorithm can be almost fully parallelized.
Our strategy for dealing with arbitrary instances is to instead inject the randomness into the dissection process itself by working over a carefully selected but random composite modulus, and to introduce explicit space--time controls into the algorithm by means of a ``bailout mechanism''.
\end{abstract}

\maketitle
\thispagestyle{plain}

\section{Introduction}

The protagonist of this paper is the \subsetsum{} problem.

\begin{definition}
  An instance $(\vec{a}, t)$ of \subsetsum{} consists of a vector
  $\vec{a} \in \Z_{\geq 0}^n$ and a target $t \in \Z_{\geq 0}$.  
  A solution of $(\vec{a}, t)$ is a vector $\vec{x} \in \{0,1\}^n$ 
  such that $\sum_{i=1}^{n} a_i x_i = t$.
\end{definition}

The problem is NP-hard (in essence, Karp's formulation of the knapsack problem~\cite{karp72}), and the fastest known algorithms take time and space that grow exponentially in $n$. We will write $T$ and $S$ for the exponential factors and omit the possible polynomial factors. 
The brute-force algorithm, with $T=2^n$ and $S=1$, was beaten four decades ago, when Horowitz and Sahni~\cite{HorowitzSahni74} gave a simple yet powerful meet-in-the-middle algorithm that achieves $T = S = 2^{n/2}$ by halving the set arbitrarily, sorting the $2^{n/2}$ subsets of each half, and then quickly scanning through the relevant pairs of subsets that could sum to the target. Some years later, Schroeppel and Shamir~\cite{SchroeppelShamir81} improved the space requirement of the algorithm to $S = 2^{n/4}$ by designing a novel way to list the half-sums in sorted order in small space. However, if allowing only polynomial space, no better than the trivial time bound of $T = 2^n$ is known. Whether the constant bases of the exponentials in these bounds can be improved is a major open problem in the area of moderately exponential algorithms~\cite{Woeginger08}.  

The difficulty of finding faster algorithms, whether in polynomial or exponential space, has motivated the study of space--time tradeoffs. From a practical point of view, large space usage is often the bottleneck of computation, and savings in space usage can have significant impact even if they come at the cost of increasing the time requirement. This is because a smaller-space algorithm can make a better use of fast cache memories and, in particular, because a smaller-space algorithm often enables easier and more efficient large-scale parallelization. Typically, one obtains a smooth space--time tradeoff by combining the fastest exponential time algorithm with the fastest polynomial space algorithm into a hybrid scheme that interpolates between the two extremes. An intriguing question is then whether one can beat the hybrid scheme at some point, that is, to get a faster algorithm at some space budget---if one can break the hybrid bound somewhere, maybe one can break it everywhere. For the \subsetsum{} problem, a hybrid scheme is obtained by first guessing some $g$ elements of the solution, and then running the algorithm of Schroeppel and Shamir for the remaining instance on $n-g$ elements. This yields $T = 2^{(n+g)/2}$ and $S = 2^{(n-g)/4}$, for any $ 0 \leq g \leq n$, and thereby the smooth tradeoff curve $S^2 T = 2^n$ for $1 \leq S \leq 2^{n/4}$. We call this the {\em Schroeppel--Shamir tradeoff}.

While the Schroeppel--Shamir tradeoff has remained unbeaten in the
usual worst-case sense, there has been remarkable recent progress in
the random-instance setting~\cite{Howgrave-GrahamJoux10,BeckerCoronJoux11,DinurDunkelmanKellerShamir12}.
In a recent result, Dinur, Dunkelman, Keller, and 
Shamir~\cite{DinurDunkelmanKellerShamir12} gave a tradeoff curve that
matches the Schroeppel--Shamir tradeoff at the extreme points $S = 1$
and $S = 2^{n/4}$ but is strictly better in between. 
The tradeoff is achieved by a novel {\em dissection} method that
recursively decomposes the problem into smaller subproblems in two
different ``dimensions'', the first dimension being the current subset
of the $n$ items, and the other dimension being (roughly speaking) the
bits of information of each item.
The algorithm of Dinur et al.~runs in space
$S = 2^{\sigma n}$ and time $T = 2^{\tau(\sigma)n}$ on random
instances ($\tau(\sigma)$ is defined momentarily).  See
Figure~\ref{fig:tradeoff} for an illustration and comparison to the
Schroeppel--Shamir tradeoff.  The tradeoff curve $\tau(\sigma)$ is
piecewise linear and determined by what Dinur et al.~call the ``magic
sequence'' $2, 4, 7, 11, 16, 22, \ldots$, obtained as evaluations of
$\rho_{\ell} = 1 + \ell(\ell+1)/2$ at $\ell = 1, 2, \ldots$.

\begin{figure}[t!]
\begin{center}
\hspace{-1.4in}
{\scriptsize 
\input{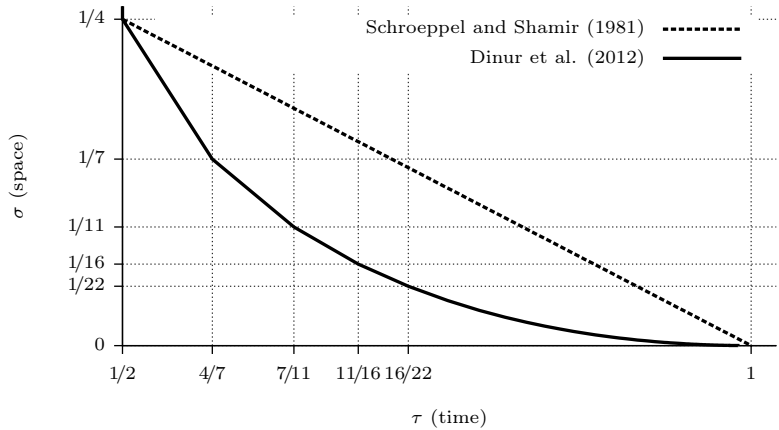}
}
\vspace*{-0.1in}
\caption{Space--time tradeoff curves for the \subsetsum{} prob\-lem~\mbox{\cite{SchroeppelShamir81,DinurDunkelmanKellerShamir12}}. The space and time requirements are $S = 2^{\sigma n}$ and $T = 2^{\tau n}$, omitting factors polynomial in the instance size $n$. }
\label{fig:tradeoff}
\end{center}
\end{figure}

\begin{definition}
  Define $\tau: (0,1] \rightarrow [0,1]$ as follows.  For $\sigma \in (0,1/2]$,
  let $\ell$ be the solution to $1/\rho_{\ell+1} < \sigma \le 1/\rho_{\ell}$.
  Then
  \begin{equation}\label{eq:tau}
	\tau(\sigma) =  1 - \frac{1}{\ell+1} - \frac{\rho_{\ell} - 2}{\ell+1}\sigma \,.
  \end{equation}
  If there is no such $\ell$, that is, if $\sigma > 1/2$, 
  define $\tau(\sigma) = 1/2$.
\end{definition}

\noindent
For example, at $\sigma = 1/8$, we have $\ell = 3$, and thereby
$\tau(\sigma) = 19/32$.   Asymptotically, when
$\sigma$ is small, $\ell$ is essentially $\sqrt{2/\sigma}$ and
$\tau(\sigma) \approx 1-\sqrt{2\sigma}$.

In this paper, we show that this space--time tradeoff result by Dinur et al.~\cite{DinurDunkelmanKellerShamir12} can be made to hold also in the worst case:
\begin{theorem}
  \label{thm:main}
For each $\sigma\in(0,1]$ there exists a randomized algorithm that solves the \subsetsum{} problem with high probability, and runs 
in $O^*(2^{\tau(\sigma) n})$ time and $O^*(2^{\sigma n})$ space.  The $O^*$ notation suppresses factors that are polynomial in $n$, and the polynomials depend on $\sigma$.
\end{theorem}
\noindent 

To the best of our knowledge, Theorem~\ref{thm:main} is the first 
improvement to the Schroeppel--Shamir tradeoff in the worst-case setting. 
Here we should remark that, in the random-instance setting, there are results
that improve on both the Schroeppel--Shamir and the Dinur et al. 
tradeoffs for certain specific choices of the space budget $S$. 
In particular, Becker et al.~give a $2^{0.72n}$ time polynomial space
algorithm and a $2^{0.291n}$ time exponential space 
algorithm~\cite{BeckerCoronJoux11}.  
A natural question that remains is whether 
these two results could be extended to the worst-case setting.
Such an extension would be a significant breakthrough
(cf.~\cite{Woeginger08}).

We also prove that the dissection algorithm lends itself to
parallelization very well.  As mentioned before, a general
guiding intuition is that algorithms that use less space can be more
efficiently parallelized.  The following theorem shows that, at least
in the case of the dissection algorithm, this intuition can be made
formal: the smaller the space budget $\sigma$ is, the closer we can
get to full parallelization.

\begin{theorem}
  \label{thm:parallel dissect}
  The algorithm of Theorem~\ref{thm:main} can be implemented to run 
  in\linebreak
  $O^*(2^{\tau(\sigma) n} / P)$ parallel time on $P$ processors each
  using $O^*(2^{\sigma n})$ space, provided $P \le
  2^{(2\tau(\sigma)-1)n}$.
\end{theorem}
\noindent 
When $\sigma$ is small, $\tau(\sigma) \approx 1 - \sqrt{2 \sigma}$ and
the bound on $P$ is roughly $2^{(\tau(\sigma) - \sqrt{2\sigma})n}$.
In other words we get a linear speedup almost all the way up to
$2^{\tau(\sigma)n}$ processors, almost full parallelization.

\subsection{Our contributions and overview of the proof.} 
At a high level, our approach will follow the Dinur et al. dissection 
framework, with essential differences in preprocessing and low-level 
implementation to alleviate the assumptions on randomness. 
In particular, while we split the instance analogously to Dinur
et al. to recover the tradeoff curve, we require more careful control
of the sub-instances beyond just subdividing the bits of the input integers
and assuming that the input is random enough to guarantee sufficient
uniformity to yield the tradeoff curve. Accordingly we find it convenient 
to revisit the derivation of the tradeoff curve and the analysis of the 
basic dissection framework to enable a self-contained exposition. 

In contrast with Dinur et al., our strategy for dealing with arbitrary 
instances is, essentially, to instead inject the required randomness 
into the dissection process itself.
We achieve this by observing that dissection can be carried out over
any algebraic structure that has a sufficiently rich family of
homomorphisms to enable us to inject entropy by selection of {\em
 random} homomorphisms, while maintaining an appropriate recursive
structure for the selected homomorphisms to facilitate dissection. For
the \subsetsum{} problem, in practice this means reduction from
$\mathbb{Z}$ to $\mathbb{Z}_M$ over a composite $M$ with a carefully
selected (but random) lattice of divisors to make sure that we can
still carry out recursive dissections analogously to Dinur et al.
This approach alone does not provide sufficient control over an
arbitrary problem instance, however.

The main obstacle is that, even with the randomness injected into the
algorithm, it is very hard to control the resource consumption of the
algorithm.  To overcome this, we add explicit resource controls into
the algorithm, by means of a somewhat cavalier
``bailout mechanism'' which causes the algorithm to simply stop when
too many partial solutions have been generated.
We set the threshold for such a bailout to be roughly the number of
partial solutions that we would have expected to see in a random
instance.  This allows us to keep its running time and space usage in
check, perfectly recovering the Dinur et al.\ tradeoff curve.  
The remaining challenge is then to prove correctness, i.e., that these 
thresholds for bailout are high enough so that no hazardous bailouts 
take place and a solution is indeed found.
To do this we perform a localized analysis on the
subtree of the recursion tree that contains a solution.  Using that
the constructed modulus $M$ contains a lot of randomness (a
consequence of the density of the primes), we can show that the
probability of a bailout in any node of this subtree is $o(1)$, 
meaning that the algorithm finds a solution with high probability.

A somewhat curious effect is that in order for our analysis to go
through, we require the original \subsetsum{} instance to have few, say
$O(1)$, distinct solutions.  In order to achieve this, we preprocess
the instance by employing routine isolation techniques in
$\mathbb{Z}_P$ {\em but implemented over} $\mathbb{Z}$ to control the
number of solutions over $\mathbb{Z}$. The reason why we need to
implement the preprocessing over $\mathbb{Z}$ rather than than work 
in the modular setting is that the dissection algorithm itself needs 
to be able to choose a modulus $M$ very carefully to deliver the tradeoff, 
and that choice is incompatible with having an extra prime $P$ 
for isolation.
This is somewhat curious because, intuitively, the more solutions an 
instance has, the easier it should be to find one.  The reason why that 
is not the case in our setting is that, further down in the recursion tree, 
when operating with a small modulus $M$, every original solution gives 
rise to many additional spurious solutions, and if there are too many original
solutions there will be too many spurious solutions.

A further property needed to support the analysis is that the numbers 
in the \subsetsum{} instance must not be too large, in particular we 
need $\log t = O(n)$.  This we can also achieve by a simple preprocessing 
step where we hash down modulo a random prime, but again with implementation 
over the integers for the same reason as above.

\subsection{Related work.}
The \subsetsum{} problem has recently been approached from related angles, with the interest in small space. Lokshtanov and Nederlof~\cite{LokshtanovNederlof10} show that the well-known pseudo-polynomial-time dynamic programming algorithm can be implemented in truly-polynomial space by algebraization. Kaski, Koivisto, and Nederlof~\cite{KaskiKoivistoNederlof12} note that the sparsity of the dynamic programming table can be exploited to speedup the computations even if allowing only polynomial space.

Smooth space--time tradeoffs have been studied also for several other hard problems. Bj\"{o}rklund et al.~\cite{BjorklundHusfeldtKaskiKoivisto08} derive a hybrid scheme for the Tutte polynomial that is a host of various counting problems on graphs. Koivisto and Parviainen~\cite{KoivistoParviainen10} consider a class of permutation problems (including, e.g., the traveling salesman problem and the feedback arc set problem) and show that a natural hybrid scheme can be beaten by a partial ordering technique. 

\subsection{Organization.}

In Section~\ref{sec:dissection algorithm} we describe the dissection
algorithm and give the main statements about its properties.  In
Section~\ref{sec:time space analysis} we show that the algorithm runs
within the desired time and space bounds.
Then, in Section~\ref{sec:correctness} we show that given a
\subsetsum{} instance with at most $O(1)$ solutions, the dissection
algorithm finds a solution.  In Section~\ref{sec:isolation}
we give a standard isolation argument reducing general \subsetsum{} to
the restricted case when there are at most $O(1)$ solutions, giving
the last puzzle piece to complete the proof of Theorem~\ref{thm:main}.
In Section~\ref{sec:parallel} we show that the algorithm lends
itself to efficient parallelization by proving
Theorem~\ref{thm:parallel dissect}.

\section{The Main Dissection Algorithm}
\label{sec:dissection algorithm}

Before describing the main algorithm, we condense 
some routine preprocessing steps into the following theorem, whose proof
we relegate to Section~\ref{sec:isolation}. 

\begin{theorem}
\label{thm:preprocessing}
  There is a polynomial-time randomized algorithm for preprocessing instances of \subsetsum{} which, given as input an instance $(\vec{a}, t)$ with $n$ elements, outputs a collection of $O(n^3)$ instances $(\vec{a}', t')$, each with $n$ elements and $\log t' = O(n)$,
  such that if $(\vec{a}, t)$ is a NO instance then so are all the 
  new instances with probability $1-o(1)$, 
  and if $(\vec{a}, t)$ is a YES instance then with 
  probability $\Omega(1)$ at least one of the new instances is 
  a YES instance with at most $O(1)$ solutions.
\end{theorem}

By applying this preprocessing we may assume
that the main algorithm receives an input  
$(\vec{a},t)$ that has $O(1)$ solutions and $\log t=O(n)$.
We then introduce a random modulus $M$ and transfer 
into a modular setting.

\begin{definition}
  An instance $(\vec{a}, t, M)$ of \modsubsetsum{} consists of a vector
  $\vec{a} \in \Z_{\geq 0}^n$, 
  a target $t \in \Z_{\geq 0}$, and a modulus $M\in\Z_{\geq 1}$.  
  A solution of $(\vec{a}, t, M)$ is a vector $\vec{x} \in \{0,1\}^n$ such
  that $\sum_{i=1}^{n} a_i x_i \equiv t\pmod M$.
\end{definition}

The reason why we transfer to the modular setting is that the
recursive dissection strategy extensively uses the fact that we have
available a sufficiently rich family of homomorphisms to split the
search space.  In particular, in the modular setting this corresponds
to the modulus $M$ being ``sufficiently divisible'' (in a sense to be
made precise later) to obtain control of the recursion.

Pseudocode for the main algorithm is given in Algorithm~\ref{alg:main}.
In addition to the modular instance $(\vec{a},t,M)$, the algorithm
accepts as further input the space parameter $\sigma\in(0,1]$. 

The key high-level idea in the algorithm is to ``meet in the middle'' by
splitting an instance of $n$ items to two sub-instances of $\alpha n$
items and $(1-\alpha)n$ items, guessing (over a smaller modulus $M'$
that divides $M$) what the sum should be after the first and before
the second sub-instance, and then recursively solving the two
sub-instances subject to the guess.
Figure~\ref{fig:algorithm-illustration} illustrates the structure of
the algorithm.

\begin{algorithm}[t]
\caption{\textsc{GenerateSolutions}$(\vec{a}, t, M, \sigma)$}
\label{alg:main}
\KwData{$(\vec{a}, t, M)$ is an $n$-element \modsubsetsum{} instance, 
$\sigma\in(0,1]$}
\KwResult{Iterates over up to $\Theta^*(2^n/M)$ solutions of $(\vec{a}, t, M)$ 
while using space $O^*(2^{\sigma n})$}
\Begin{
    \If{$\sigma \ge 1/4$}{
      Report up to $\Theta^*(2^n/M)$ solutions using the Shroeppel-Shamir algorithm\;
      \Return\;
    }
    Choose $\alpha \in (0,1), \beta \in (0,1)$ appropriately (according to Theorem~\ref{thm:time bound}) based on $\sigma$\;
    Let $M'$ be a factor of $M$ of magnitude $\Theta(2^{\beta n})$\;
    \For{$s'=0,1,\ldots,M'-1$     \label{step:s-loop} }{ 
      Allocate an empty lookup table\;
      Let $\vec{l} = (a_1, a_2, \ldots, a_{\alpha n})$ be the first $\alpha n$ items of $\vec{a}$\;
      Let $\vec{r} = (a_{\alpha n + 1}, a_{\alpha n + 2}, \ldots, a_{n})$ be the remaining $(1-\alpha) n$ items of $\vec{a}$\;
      \For{$\vec{y} \in \textsc{GenerateSolutions}(\vec{l}, s', M', \frac{\sigma}{\alpha})$}{ \label{step:firstrec}
        Let $s = \sum_{i=1}^{\alpha n} a_i y_i \bmod M$\;
        Store $[s \rightarrow \vec{y}]$ in the lookup table\;
      }
      \For{$\vec{z} \in \textsc{GenerateSolutions}(\vec{r}, t-s', M', \frac{\sigma}{1-\alpha})$}{ \label{step:secondrec}
        Let $s = t-\sum_{i=\alpha n+1}^{n} a_iz_i \bmod M$\;
        \ForEach{$[s \rightarrow \vec{y}]$ in the lookup table}{
          \label{step:comb-start}
          Report solution $\vec{x}=(\vec{y}, \vec{z})$\;
          \If{at least $\Theta^*(2^n/M)$ solutions reported}{
            Stop iteration and {\bf return}\; \label{step:bailout}
          }
          \label{step:comb-end}
        }
      }
      Release the lookup table\;
    }
}
\end{algorithm}

\begin{figure}
\begin{center}
\includegraphics[width=0.7\linewidth]{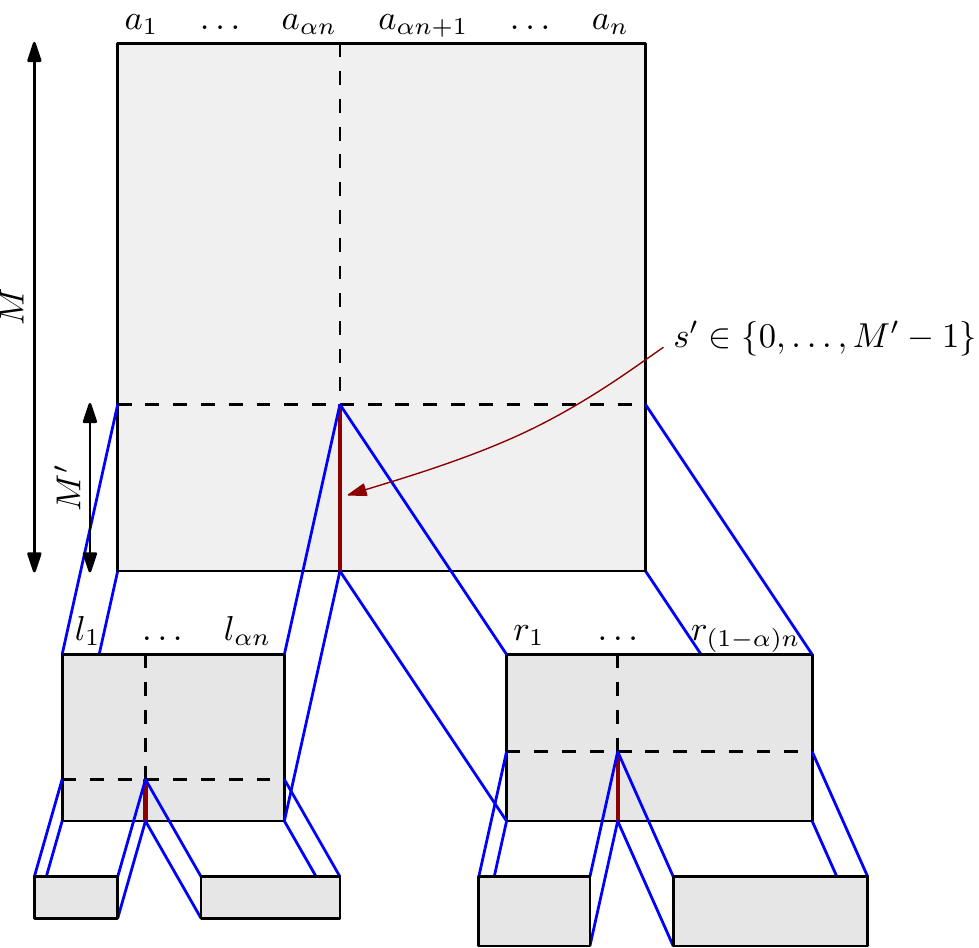}
\caption{Illustration of the recursive dissections made by the algorithm.} 
\label{fig:algorithm-illustration}
\end{center}
\end{figure}

We continue with some further high-level remarks.

\begin{enumerate}
\item In the algorithm, two key parameters $\alpha$ and
$\beta$ are chosen, which control how the \modsubsetsum{} instance is
subdivided for the recursive calls.  The precise choice of these
parameters is given in Theorem~\ref{thm:time bound} below,
but at this point the reader is
encouraged to simply think of them as some parameters which should be
chosen appropriately so as to optimize running time.

\item The algorithm also chooses a factor $M'$ of $M$ 
  such that $M'=\Theta(2^{\beta n})$. The existence of sufficient
  factors at all levels of recursion is established in 
  Section~\ref{sec:correctness}.

\item The algorithm should be viewed as an {\em iterator} over
  solutions.  In other words, the algorithm has an internal state, 
  and a {\em next item} functionality that we tacitly use by 
  writing a for-loop over all solutions generated by the algorithm,
  which should be interpreted as a short-hand for repeatedly asking 
  the iterator for the next item.

\item The algorithm uses a ``bailout mechanism'' to control the
  running time and space usage. Namely, each recursive call 
  will bail out after $\Theta^*(2^n/M)$ solutions are reported.
  (The precise bailout bound has a further multiplicative factor 
  polynomial in $n$ that depends on the top-level value of $\sigma$.)
  A preliminary intuition for the bound is that this is what one would
  expect to receive in a particular congruence class modulo $M$
  if the $2^n$ possible sums are randomly placed into the congruence classes.
\end{enumerate}

As a warmup to the analysis, let us first observe that, 
if we did not have the bailout step in line~\ref{step:bailout},
correctness of the algorithm would be more or less immediate: 
for any solution $\vec{x}$ of $(\vec{a}, t, M)$, let 
$s = \sum_{i=1}^{\alpha n} a_i x_i \bmod M$.  Then, when $s'=s \bmod{M'}$ 
in the outer for-loop (line \ref{step:s-loop}), by an inductive argument 
we will find $\vec{y}$ and $\vec{z}$ in the two separate recursive branches 
and join the two partial solutions to form $\vec{x}$.

The challenge, of course, is that without the bailout mechanism 
we lack control over the resource consumption of the algorithm.
Even though we have applied isolation to guarantee that there are not
too many solutions of the top-level instance $(\vec{a},t)$, it may be that
some branches of the recursion generate a huge number of solutions,
affecting both running time and space (since we store partial solutions 
in a lookup table).

Let us then proceed to analyzing the algorithm with the bailout
mechanism in place.  The two main claims are as follows.

\begin{theorem}
  \label{thm:time bound}
  Given a space budget $\sigma\in(0,1]$ and $M \ge 2^n$, if in each recursive step of Algorithm~\ref{alg:main} the parameters $\alpha$ and $\beta$ are chosen as
  \begin{align}
    \label{eqn:alpha beta choice}
    \alpha & = 1 - \tau(\sigma) &\text{and}&& \beta & = 1 - \tau(\sigma) - \sigma\,,
  \end{align}
  then the algorithm runs in $O^*(2^{\tau(\sigma) n})$ time and 
  $O^*(2^{\sigma n})$ space.
\end{theorem}

\begin{theorem}
  \label{thm:no critical bailouts}
  For every $\sigma\in(0,1]$  
  there is a randomized algorithm that runs in time polynomial
  in $n$ and chooses a top-level modulus $M\geq 2^n$ so that 
  Algorithm~\ref{alg:main} reports a solution of the non-modular
  instance $(\vec{a},t)$ with high probability over the choices of $M$, 
  assuming that at least one and at most $O(1)$ solutions exist
  and that $\log t = O(n)$.
\end{theorem}
\noindent We prove Theorem~\ref{thm:time bound} in Section~\ref{sec:time space
  analysis} and Theorem~\ref{thm:no critical bailouts} in
Section~\ref{sec:correctness}.

Let us however here briefly discuss the specific choice of $\alpha$
and $\beta$ in Theorem~\ref{thm:time bound}.  We arrived at
\eqref{eqn:alpha beta choice} by analyzing the recurrence relation
describing the running time of Algorithm~\ref{alg:main}.
Unfortunately this recurrence in its full form is somewhat
complicated, and our process of coming up with \eqref{eqn:alpha beta
  choice} involved a certain amount of experimenting and guesswork.
We do have some guiding (non-formal) intuition which might be
instructive:
\begin{enumerate}
\item One needs to make sure that $\alpha - \beta \le \sigma$.  This
  is because for a random instance, the left subinstance is expected
  to have roughly $2^{(\alpha - \beta)n}$ solutions, and since we need 
  to store these there had better be at most $2^{\sigma n}$ of them.
\item Since $\beta \ge \alpha - \sigma$ and $\beta$ has a very direct
  impact on running time (due to the $2^{\beta n}$ time outer loop),
  one will typically want to set $\alpha$ relatively small.  The
  tension here is of course that the smaller $\alpha$ becomes, the
  larger $1-\alpha$ (that is, the size of the right subinstance) becomes.
\item Given this tension, setting $\alpha - \beta = \sigma$ is
  natural.
\end{enumerate}

So in an intuitive sense, the bottleneck for space comes from the left
subinstance, or rather the need to store all the solutions found for the left subinstance (this is not technically true since we give the right
subinstance $2^{\sigma n}$ space allowance as well), whereas the
bottleneck for time comes from the right subinstance, which tends to
be much larger than the left one.

\section{Analysis of Running Time and Space Usage}
\label{sec:time space analysis}

In this section we prove Theorem~\ref{thm:time bound} giving the
running time upper bound on the dissection algorithm.  For this, it is
convenient to define the following function, which is less explicit
than $\tau$ but more naturally captures the running time of the
algorithm.

\begin{definition}
  \label{def:F recurrence}
  Define $F: (0,1] \rightarrow (0,1)$ by 
  the following recurrence for $\sigma < 1/4$:
  \begin{equation}
  \label{eqn:F defn}
  F(\sigma) = \beta + \max \Big\{\alpha F\Big(\frac{\sigma}{\alpha}\Big), (1-\alpha) F\Big(\frac{\sigma}{1-\alpha}\Big)\Big\}\,,
  \end{equation}
  where $\alpha = 1 - \tau(\sigma)$ and $\beta = \alpha - \sigma$.
  The base case is $F(\sigma) = 1/2$ for $\sigma \ge 1/4$.
\end{definition}

To analyze the running time of the dissection algorithm, let us first
define a ``dummy'' version of Algorithm~\ref{alg:main}, 
given as Algorithm~\ref{alg:dummy}.
The dummy version is a bare bones version of Algorithm~\ref{alg:main}
which generates the same recursion tree.

\begin{algorithm}[t]
  \caption{\textsc{DummyDissection}$(n, \sigma)$}
  \label{alg:dummy}
  \KwData{$\sigma\in(0,1]$}
    \Begin{
        \If{$\sigma \ge 1/4$}{
          Run for $2^{n/2}$ steps\;
          \Return\;
        }
        Let $\alpha = 1 - \tau(\sigma)$, $\beta = \alpha - \sigma$\;
        \For{$2^{\beta n}$ steps}{
          $\textsc{DummyDissection}(\alpha n, \sigma / \alpha)$\;
          $\textsc{DummyDissection}((1-\alpha) n, \sigma / (1-\alpha))$\;
        }
      }
\end{algorithm}
The following lemma is immediate from the definition of $F(\sigma)$.

\begin{lemma}
  \label{lemma:time bound for dummy dissection}
  Algorithm~\ref{alg:dummy} runs in 
  $O^*(2^{F(\sigma) n})$ time on input $(n, \sigma)$.
\end{lemma}

Next, we can relate the running time of the dummy algorithm to the
running time of the actual algorithm.  Ignoring polynomial factors 
such as those arising from updating the lookup table, the only
time-consuming step of Algorithm~\ref{alg:main} that we have omitted
in Algorithm~\ref{alg:dummy} is the combination loop in
steps~\ref{step:comb-start} to~\ref{step:comb-end}. 
The total amount of time spent in this loop in any fixed recursive
call is, by virtue of step~\ref{step:bailout}, at most $O^*(2^n/M)$.
So if $M \ge 2^{(1-F(\sigma))n}$ then this time is dominated by 
the run time from the recursive calls. In other words:

\begin{lemma}
  \label{lemma:dummy dissection simulates real}
  Consider running Algorithm~\ref{alg:main} on input $(\vec{a}, t, M,
  \sigma)$.  If in every recursive call made it holds that $M \ge
  2^{(1-F(\sigma))n}$ then the running time is within a polynomial
  factor of the running time of Algorithm~\ref{alg:dummy} on input
  $(n, \sigma)$, that is, at most $O^*(2^{F(\sigma)n})$.
\end{lemma}

The next key piece is the following lemma, stating that the function
$F$ is nothing more than a reformulation of $\tau(\sigma)$.
We defer the proof to Section~\ref{sec:tau leq F proof}.

\begin{lemma}
  \label{lemma:tau leq F}
  For every $\sigma \in (0, 1]$ it holds that $F(\sigma) = \tau(\sigma)$.
\end{lemma}

Equipped with this lemma, we are in good shape 
to prove Theorem~\ref{thm:time bound}.

\begin{reptheorem}{thm:time bound}
  Given a space budget $\sigma\in(0,1]$ and $M \ge 2^n$, if in each recursive step of Algorithm~\ref{alg:main} the parameters $\alpha$ and $\beta$ are chosen as
  \begin{align}
    \alpha & = 1 - \tau(\sigma) &\text{and}&& \beta & = 1 - \tau(\sigma) - \sigma\,,
  \end{align}
  then the algorithm runs in $O^*(2^{\tau(\sigma) n})$ time and 
  $O^*(2^{\sigma n})$ space.
\end{reptheorem}

\begin{proof}[Proof of Theorem~\ref{thm:time bound}]
  Let us start with space usage.  There are three items to bound: (1)
  the space usage in the left branch (step~\ref{step:firstrec}), (2)
  the space usage in the right branch (step~\ref{step:secondrec}), and
  (3) the total number of solutions found in the left branch (as these
  are all stored in a lookup table).  For (1), the subinstance
  $(\vec{l},s',M')$ has $\alpha n$ items and has a space budget of
  $\sigma / \alpha$, so by an inductive argument it uses space
  $O(2^{\frac{\sigma}{\alpha} \alpha n}) = O(2^{\sigma n})$.  The case
  for (2) is analogous.  It remains to bound (3), which is clearly
  bounded by the number of solutions found in the recursive
  step~\ref{step:firstrec}.  However, by construction, this is 
  (up to a suppressed factor polynomial in $n$)
  at most $2^{\alpha n} / M' = O(2^{(\alpha - \beta)n}) = O(2^{\sigma n})$.

  We thus conclude that the total space usage of the algorithm is
  bounded by $O^*(d 2^{\sigma n})$ where $d$ is the recursion depth,
  which is $O(1)$ by Lemma~\ref{lemma:dissection depth}.

  Let us turn to time usage.
  First, to apply Lemma~\ref{lemma:dummy
    dissection simulates real}, we need to make sure that we always
  have $M \ge 2^{(1-F(\sigma))n} = 2^{(1-\tau(\sigma))n}$ in every
  recursive call.  In the top level call this is true since $M \ge
  2^n$.  Suppose (inductively) that it is true in some recursive call,
  and let us prove that it holds for both left- and right-recursive calls.
	We refer to the respective values of the parameters by adding subscripts $l$ and $r$.

  In a left-recursive call, we have $n_l = \alpha n$, $M_l = 2^{\beta
    n}$, and $\sigma_l = \sigma / \alpha$.  We thus need $2^{\beta n}
  \ge 2^{(1-\tau(\sigma/\alpha))\alpha n}$.  Noting that
  $1-\tau(\sigma/\alpha) \le 1/2$ and that $\beta \ge \alpha/2$ (this
  is equivalent to $\tau(\sigma) < 1-2\sigma$), we see that $M_l$ is
  sufficiently large.

  In a right-recursive call, we have $n_r = (1-\alpha) n = \tau(\sigma) n$,
  $M_r = 2^{\beta n}$, and $\sigma_r = \sigma / (1-\alpha) =
  \sigma/\tau(\sigma)$.  By Proposition~\ref{prop:tau recurrence},
  we have $1-\tau(\sigma_r) = (1 - \sigma -
    \tau(\sigma))/\tau(\sigma) = \beta/\tau(\sigma)$, from which we
  conclude that $M_r = 2^{(1-\tau(\sigma_r))n_r}$.

  Thus the conditions of Lemma~\ref{lemma:dummy
    dissection simulates real} are satisfied, and the running time bound
  of $O^*(2^{\tau(\sigma) n})$ for Algorithm~\ref{alg:main} is a direct
  consequence of Lemmata~\ref{lemma:time bound for dummy dissection}, \ref{lemma:dummy dissection simulates real}, and~\ref{lemma:tau leq F}.
\end{proof}

\subsection{Proof of Lemma~\ref{lemma:tau leq F}}
\label{sec:tau leq F proof}

We first prove some useful properties of the $\tau$ function.

\begin{proposition}
  \label{prop:mu-prime}
  The map $\sigma \mapsto \sigma/\tau(\sigma)$ is increasing in $\sigma \in (0, 1]$. Furthermore,
    for $\sigma = 1/\rho_{\ell+1}$, we have $\sigma/\tau(\sigma) = 1/\rho_{\ell}$.
\end{proposition}
\begin{proof}
  Let $\sigma \in (0, 1]$, and let $\sigma' = \sigma/\tau(\sigma)$.
    If $\sigma > 1/2$, then $\tau(\sigma) = 1/2$, and thus $\sigma' =  2\sigma$ is increasing in $\sigma$.
    Otherwise $1/\rho_{\ell+1} < \sigma \leq 1/\rho_{\ell}$ for some $\ell \geq 1$, and $\tau(\sigma) = (\ell - (\rho_{\ell}-2)\sigma)/(\ell+1)$. Thus  
    \begin{equation*}
      \frac{1}{\sigma'} = \frac{\tau(\sigma)}{\sigma} = \frac{\ell - (\rho_{\ell} - 2)\sigma}{(\ell+ 1)\sigma} = \frac{\ell/\sigma - \rho_{\ell} + 2}{\ell+1}\,,
    \end{equation*}
    from which it follows that $\sigma'$ is increasing in $\sigma$ in the interval $(1/\rho_{\ell+1}, 1/\rho_{\ell}]$. 

      Suppose $\sigma = 1/\rho_{\ell+1}$. Use first $\rho_{\ell+1} = \rho_{\ell} + \ell +1$ and then $\ell(\ell+1) = 2(\rho_{\ell} -1)$ to obtain 
      \[
      \frac{1}{\sigma'} = 
      \frac{\ell(\rho_{\ell}+\ell+1) - \rho_{\ell} + 2}{\ell+1} = \frac{(\ell-1)\rho_{\ell} + 2(\rho_{\ell}-1) + 2}{\ell+1} = \rho_{\ell}\,. \qedhere
      \]
\end{proof}

\begin{proposition}
  \label{prop:tau recurrence}
  Let $\sigma \in (0, 1]$. If $\sigma > 1/2$, then $\tau(\sigma) = 1/2$, and otherwise
	\[
	\tau(\sigma) = \frac{1-\sigma}{2-\tau(\sigma/\tau(\sigma))}\,.
	\] 
\end{proposition}

\begin{proof}
  The case $\sigma > 1/2$ is obvious.  Fix $\sigma \le 1/2$ and $\ell
  \ge 1$ such that $1/\rho_{\ell+1} < \sigma \leq 1/\rho_{\ell}$ and let
  $\sigma' = \sigma/\tau(\sigma)$.  By Proposition~\ref{prop:mu-prime}
  we have that $1/\rho_{\ell} < \sigma' \le 1/\rho_{\ell-1}$.
  Using $\tau(\sigma') = (\ell-1 - (\rho_{\ell-1}-2)\sigma')/\ell$ we obtain
  \[2-\tau(\sigma') = \frac{\ell+1 + (\rho_{\ell-1}-2)\sigma'}{\ell}\,.
  \]
  Plugging in
  $\sigma' = \sigma/\tau(\sigma)$ and using $\rho_{\ell-1} =
  \rho_{\ell}-\ell$ gives
  \begin{equation}
    \label{eqn:2 minus tau sigma prime}
    2-\tau(\sigma') = \frac{(\ell+1) \tau(\sigma) + (\rho_{\ell}-\ell-2)\sigma}{\ell\tau(\sigma)}\,.
  \end{equation}
  As $\tau(\sigma) = (\ell - (\rho_{\ell}-2)\sigma)/(\ell+1)$, the numerator
  of this expression equals
  \[
  \ell - (\rho_{\ell}-2)\sigma + (\rho_{\ell}-\ell-2)\sigma = \ell(1-\sigma)\,.
  \]
  Plugging this into \eqref{eqn:2 minus tau sigma prime} we conclude that
  \[
  2- \tau(\sigma') = \frac{1-\sigma}{\tau(\sigma)}\,,
  \]
  which is a simple rearrangement of the desired conclusion.
\end{proof}

We are now ready to prove Lemma~\ref{lemma:tau leq F}.

\begin{replemma}{lemma:tau leq F}
  For every $\sigma \in (0, 1]$ it holds that $F(\sigma) = \tau(\sigma)$.
\end{replemma}

\begin{proof}[Proof of Lemma~\ref{lemma:tau leq F}]
  The proof is by induction on the value of $\ell$ such that $1/\rho_{\ell+1} <
  \sigma \le 1/\rho_{\ell}$.  The base case, $\sigma \ge 1/4$ (that is, $\ell \leq 1$) is
  clear from the definitions.

  For the induction step, fix some value of $\ell \ge 2$, and assume
  that $F(\sigma') = \tau(\sigma')$ for all $\sigma' > 1/\rho_{\ell}$.  
  We need to show that for any $\sigma$ in the
  interval $[1/\rho_{\ell+1}, 1/\rho_{\ell})$, it holds that $F(\sigma)
    = \tau(\sigma)$. To this end, we set $\alpha = 1-\tau(\sigma)$ and $\beta = 1-\tau(\sigma)-\sigma$, 
    and show that the two options in the $\max$ in
    \eqref{eqn:F defn} are bounded by $\tau(\sigma)$, one with equality.

    Consider first the second option. Set $\sigma' = \sigma/(1-\alpha) = \sigma/\tau(\sigma)$. 
    By Proposition~\ref{prop:mu-prime}, we have $\sigma' > 1/\rho_{\ell}$.
    Thus, by the induction hypothesis we have $F(\sigma') =
    \tau(\sigma')$, and hence the second option in \eqref{eqn:F defn} equals
    \[
    \beta + (1-\alpha)\tau(\sigma') = 1-\tau(\sigma)-\sigma + \tau(\sigma)\tau(\sigma/\tau(\sigma)) = \tau(\sigma)\,,
    \]
    where the last step is an application of Proposition~\ref{prop:tau recurrence}.

    Consider then the first option.  Let $\sigma'' =
    \sigma/\alpha$ be the value passed to $F$ in this branch.  It is
    easy to check that $\sigma'' \ge \sigma' > 1/\rho_{\ell}$. So the induction
    hypothesis applies, and we get an upper bound of
    \[
    \beta + \alpha \tau(\sigma'') < \beta + (1-\alpha)\tau(\sigma') \le \tau(\sigma)\,.
    \]
    The first step uses $\tau(\sigma) \ge 1/2$ (yielding $\alpha < 1/2$) and
    the monotonicity of $\tau$, and the last step uses the bound on the second option.
\end{proof}

\section{Choice of Modulus and Analysis of Correctness}
\label{sec:correctness}

In this section we prove Theorem~\ref{thm:no critical bailouts}, giving the
correctness of the dissection algorithm.

\subsection{The dissection tree.}
Now that we have the choice of $\alpha$ and $\beta$ in
Algorithm~\ref{alg:main}, we can look more closely at the recursive
structure of the algorithm.  To this end, we make the following
definition.

\newcommand{\dissecttree}{\mathcal{DT}}

\begin{definition}[Dissection tree]
  For $\sigma \in (0,1]$, the \emph{dissection tree}
  $\dissecttree(\sigma)$ is the ordered binary tree defined as follows.
  If $\sigma \ge 1/4$ then $\dissecttree(\sigma)$ is a single node.
  Otherwise, let $\alpha = 1 - \tau(\sigma)$.  The left child of
  $\dissecttree(\sigma)$ is $\dissecttree(\sigma/\alpha)$, and the right
  child of $\dissecttree(\sigma)$ is $\dissecttree(\sigma/(1-\alpha))$.
\end{definition}

\begin{figure}
\begin{center}
\includegraphics[width=\linewidth]{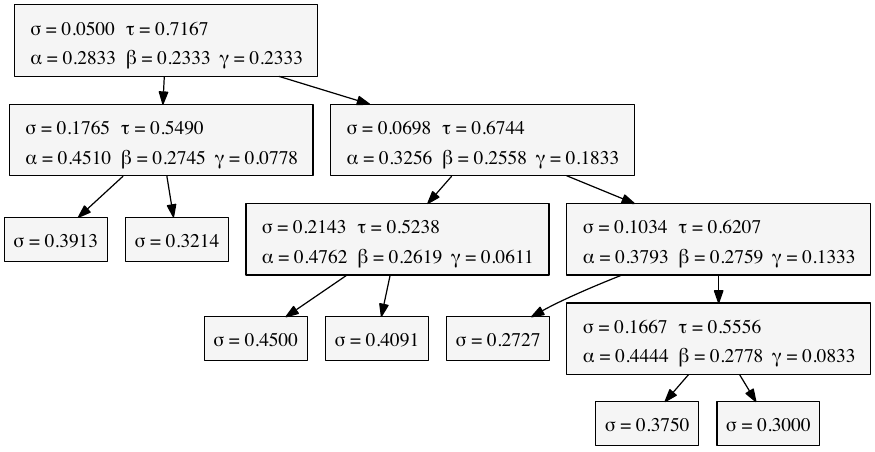}
\caption{The dissection tree $\dissecttree(0.05)$. For each internal node $v$, we display the parameters $\sigma_v, \tau_v = \tau(\sigma_v), \alpha_v, \beta_v, \gamma_v$ as defined in Section~\ref{sec:correctness}.}
\label{fig:dissecttree-illustration}
\end{center}
\end{figure}

Figure~\ref{fig:dissecttree-illustration} shows
$\dissecttree(0.05)$.  The dissection tree captures the essence of the
recursive behaviour of the dissection algorithm when being run with
parameter $\sigma$.  The actual recursion tree of the dissection
algorithm is huge due to the for-loop over $s'$ in line
\ref{step:s-loop}, but if we consider a fixed choice of $s'$ in every
recursive step then the recursion tree of the algorithm becomes
identical to the corresponding dissection tree.

\begin{lemma}
  \label{lemma:dissection depth}
  The recursion depth of Algorithm~\ref{alg:main} is the height of
  $\dissecttree(\sigma)$. In particular, the recursion depth is a
  constant that depends only on $\sigma$.
\end{lemma}

We now describe how to choose
{\em a priori} a random $M$ that is ``sufficiently divisible'' for 
the algorithm's desires, and to show correctness of the algorithm.

Fix a choice of the top-level value $\sigma\in(0,1]$. 
Consider the corresponding dissection tree $\dissecttree(\sigma)$.  
For each node $v$ of $\dissecttree(\sigma)$, write $\sigma_v$ for 
the associated $\sigma$ value.  For an internal node $v$ let us also 
define 
$\alpha_v = 1-\tau(\sigma_v)$ and $\beta_v = 1 - \sigma_v - \tau(\sigma_v)$.  
In other words, if $v_1$ and $v_2$ are the two child nodes 
of $v$, then $\sigma_{v_1} = \sigma_v/\alpha_v$ and 
$\sigma_{v_2} = \sigma_v/(1-\alpha_v)$.
Finally, define $\gamma_v = \beta_v \cdot \sigma / \sigma_v$.

Observe that each recursive call made by Algorithm~\ref{alg:main}
is associated with a unique internal node $v$ of the dissection tree
$\dissecttree(\sigma)$.

\begin{lemma}
  Each recursive call associated with an internal node $v$ 
  requires a factor $M'$ of magnitude $\Theta^*(2^{\gamma_v n})$.
\end{lemma}

\begin{proof}
  Telescope a product of the ratio $\sigma_p / \sigma_u$ 
  for a node $u$ and its parent $p$ along the path from $v$ 
  to the root node. Each such $\sigma_p/\sigma_u$ is
  either $\alpha_u$ or $1-\alpha_u$ depending on whether it is 
  a left branch or right branch---precisely the factor by which $n$
  decreases.
\end{proof}

\subsection{Choosing the modulus.}
The following lemma contains the algorithm that chooses the random modulus.

\begin{lemma}
  \label{lem:m-prime}
  For every $\sigma\in(0,1]$  
  there exists a randomized algorithm that, given integers
  $n$ and $b = O(n)$ as input,
  runs in time polynomial in $n$ and outputs for each internal node $v
  \in \dissecttree(\sigma)$ random moduli $M_v$, $M_v'$ such that, for the
  root node $r \in \dissecttree(\sigma)$, $M_r \ge 2^b$,
  and furthermore for every internal node $v$:
  \begin{enumerate}
  \item $M_v'$ is of magnitude $\Theta(2^{\gamma_v n})$,
    \label{lem:m-prime:item1}
  \item $M_v = M_p'$, where $p$ is the parent of $v$,
    \label{lem:m-prime:item2}
  \item $M_v'$ divides $M_v$, and
    \label{lem:m-prime:item3}
  \item for any fixed integer $1 \le Z \le 2^b$, the probability that $M_v'$
    divides $Z$ is $O^*(1/M_v')$.
    \label{lem:m-prime:item4}
  \end{enumerate}
\end{lemma}

\begin{proof}
  Let $0 < \lambda_1 < \lambda_2 < \cdots < \lambda_k$ be the set of
  distinct values of $\gamma_v$ ordered by value, and let $\delta_i =
  \lambda_i - \lambda_{i-1}$ be their successive differences (where we
  set $\lambda_0 = 0$ so that $\delta_1 = \lambda_1$).
  Since $\dissecttree(\sigma)$ depends only $\sigma$ and not on $n$,
  we have $k=O(1)$. For each $1 \le i \le k$ independently, let $p_i$ be a 
  uniform random prime from the interval 
  $[2^{\delta_in},2 \cdot 2^{\delta_in}]$.

  For a node $v$ such that $\gamma_v = \lambda_j$, let $M_v' =
  \prod_{i=1}^j p_j$.  Condition \ref{lem:m-prime:item1} then holds by
  construction.
  The values of $M_v$ are the determined for all nodes except the root
  through condition \ref{lem:m-prime:item2}; for the root node $r$ we
  set $M_r = p_0 M_r'$, where $p_0$ is a random prime of magnitude
  $2^{\Theta(n)}$ to make sure that $M_r \ge 2^b$.

  To prove condition \ref{lem:m-prime:item3} note that for any node
  $v$ with parent $p$, we need to prove that $M_v'$ divides $M_p'$.
  Let $j_v$ be such that $\lambda_{j_v} = \gamma_v$ and $j_p$ such
  that $\lambda_{j_p} = \gamma_p$.  Noting that the value of
  $\gamma_v$ decreases as one goes down the dissection tree, it then
  holds that $j_v < j_p$, from which it follows that $M_v' =
  \prod_{i=1}^{j_v} p_i$ divides $M_p' = \prod_{i=1}^{j_p} p_i$.

  Finally, for condition \ref{lem:m-prime:item4}, again let $j$ be
  such that $\lambda_j = \gamma_v$, and observe that in order for $Z$
  to divide $M_v'$ it must have all the factors $p_1, p_2, \ldots, p_j$.
  For each $1 \le i \le j$, $Z$ can have at most $\frac{\log_2
    Z}{\delta_i n} = O(1)$ different factors between $2^{\delta_i n}$
  and $2 \cdot 2^{\delta_i n}$, so by the Prime Number Theorem, the
  probability that $p_i$ divides $Z$ is at most $O(n 2^{-\delta_i
    n})$.  As the $p_i$'s are chosen independently the probability
  that $Z$ divides all of $p_1, p_2, \ldots, p_j$ (that is, $M_v'$) is 
  $O(n^j 2^{-(\delta_1 + \delta_2 + \ldots + \delta_j) n})
  = O(n^k 2^{-\gamma_v n}) = O^*(1/M_v')$, as desired.
\end{proof}

\subsection{Proof of correctness.} 
We are now ready to prove the correctness of the entire algorithm,
assuming preprocessing and isolation has been carried out.

\begin{reptheorem}{thm:no critical bailouts}
  For every $\sigma\in(0,1]$  
  there is a randomized algorithm that runs in time polynomial
  in $n$ and chooses a top-level modulus $M\geq 2^n$ so that 
  Algorithm~\ref{alg:main} reports a solution of the non-modular
  instance $(\vec{a},t)$ with high probability over the choices of $M$, 
  assuming that at least one and at most $O(1)$ solutions exist
  and that $\log t = O(n)$.
\end{reptheorem}

\begin{proof}
  The modulus $M$ is chosen using Lemma~\ref{lem:m-prime}, with $b$
  set to $\max\{n, \log nt\} = \Theta(n)$.  Specifically, it is chosen
  as $M_r$ for the root node $r$ of $\dissecttree(\sigma)$.

  Fix a solution $\vec{x}^*$ of $(\vec{a}, t)$, that is, $\sum_{i=1}^n
  a_ix^*_i=t$.  (Note that this is an equality over the integers and
  not a modular congruence.)  By assumption such an $\vec{x}^*$ exists
  and there are at most $O(1)$ choices. 

  If $\sigma\geq 1/2$, the top level recursive call executes the 
  Schroeppel--Shamir algorithm and a solution will be discovered. 
  So suppose that $\sigma\in(0,1/4)$.

  For an internal node $v \in \dissecttree(\sigma)$ consider 
  a recursive call associated with $v$, and let $L_v \subseteq [n]$ 
  (resp.~$R_v\subseteq [n]$) be the set of $\alpha_v n_v$
  (resp.~$(1-\alpha_v)n_v$) indices of the items that are passed to
  the left (resp.~right) recursive subtree of $v$. Note that these
  indices are with respect to the top-level instance, and that they do
  not depend on the choices of $s'$ made in the recursive calls. Let
  $s'_v \in \{0, \ldots, M'_v\}$ be the choice of $s'$ that could lead
  to the discovery of $\vec{x}^*$, in other words $s'_v = \sum_{i \in L_v}
  a_i x_i^* \bmod M_v'$. Let $I_v=L_v\cup R_v$. 

  For a leaf node $v\in\dissecttree(\sigma)$ and its parent $p$, 
  define $I_v=L_p$ if $v$ is a left child of $p$, and $I_v=R_p$ 
  if $v$ is a right child of $p$.
  
  We now restrict our attention to the part of the recursion tree
  associated with the discovery of $\vec{x}^*$, or in other words, the
  recursion tree obtained by fixing the value of $s'$ to $s'_v$ in
  each recursive step, rather than trying all possibilities.  This
  restricted recursion tree is simply $\dissecttree(\sigma)$.  Thus
  the set of items $\vec{a}_v = (a_i)_{i \in I_v}$ and the target
  $t_v$ associated with $v$ is well-defined for all $v\in\dissecttree(\sigma)$.
  
  Denote by $B(v)$ the event that $(\vec{a}_v, t_v, M_v)$ has more
  than $O^*(2^{n_v}/M_v)$ solutions. 
  Clearly, if $B(v)$ does not happen then there can not be a bailout
  at node $v$.\footnote{The converse is not true though: it can be that $B(v)$
  happens but a bailout happens in one (or both) of the two subtrees of $v$,
  causing the recursive call associated with node $v$ to not find all the solutions to $(\vec{a}_v, t_v, M_v)$ and thereby not bail out.}
  We will show that $\cup_{v\in \dissecttree(\sigma)}B(v)$ happens 
  with probability $o(1)$ over the choices of $\{M_v, M_v'\}$ from Lemma~\ref{lem:m-prime}, which thus implies that $\vec{x}^*$ is 
  discovered with probability $1-o(1)$. Because $\dissecttree(\sigma)$
  has $O(1)$ nodes, by the union bound it suffices to show 
  that $\Pr[B(v)] = o(1)$ for every $v \in \dissecttree(\sigma)$.

  Consider an arbitrary node $v\in\dissecttree(\sigma)$.  
  There are two types of solutions $\vec{x}_v$ of the instance 
  $(\vec{a}_v, t_v, M_v)$ associated with $v$.

  First, a vector $\vec{x}_v \in\{0,1\}^{n_v}$ is a solution if
  $\sum_{i=1}^{n_v} a_{v,i}x_{v,i}= \sum_{i \in I_v}
  a_ix^*_i$.  (Note that this is an equality over the integers, not a
  modular congruence.)  Because there are at most $O(1)$ 
  solutions to the top-level instance, there are at most $O(1)$ 
  such vectors $\vec{x}_v$.  Indeed, otherwise we would have more
  than $O(1)$ solutions of the top level instance, a contradiction.

  Second, consider a vector $\vec{x}_v \in\{0,1\}^{n_v}$ such
  that $\sum_{i=1}^{n_v} a_{v,i}x_{v,i} \neq \sum_{i\in I_v}
  a_ix^*_i$ (over the integers).  Let $Z=|\sum_{i=1}^{n_v}
  a_{v,i}x_{v,i}-\sum_{i\in I_v} a_ix^*_i|\neq 0$.  Such 
  a vector $\vec{x}_v$ is a solution of $(\vec{a}_v, t_v, M_v)$ 
  only if $M_v$ divides $Z$. Since $\log
  t=O(n)$ and $1\leq Z\leq nt$, by Lemma~\ref{lem:m-prime},
  item~\ref{lem:m-prime:item4} we have that $Z$ is divisible by
  $M_v$ with probability $O^*(1/M_v)$.

  From the two cases it follows that the expected number of solutions
  $\vec{x}_v$ of $(\vec{a}_v, t_v, M_v)$ is $E = O^*(2^{n_v}/M_v)$.
  (We remark that the degree in the suppressed polynomial depends on
  $\sigma$ but not on $n$.)
  Setting the precise bailout threshold to $n \cdot E$, we then have
  by Markov's inequality that  
  $\Pr[B(v)] = \Pr[\text{\#solutions $\vec{x}_v$} > n E] < 1/n = o(1)$,
  as desired. Since $v$ was arbitrary, we are done.
\end{proof}

\section{Preprocessing and Isolation}
\label{sec:isolation}

This section proves Theorem~\ref{thm:preprocessing} using standard
isolation techniques.

\begin{reptheorem}{thm:preprocessing}
  There is a polynomial-time randomized algorithm for preprocessing instances of \subsetsum{} which, given as input an instance $(\vec{a}, t)$ with $n$ elements, outputs a collection of $O(n^3)$ instances $(\vec{a}', t')$, each with $n$ elements and $\log t' = O(n)$,
  such that if $(\vec{a}, t)$ is a NO instance then so are all the 
  new instances with probability $1-o(1)$, 
  and if $(\vec{a}, t)$ is a YES instance then with 
  probability $\Omega(1)$ at least one of the new instances is 
  a YES instance with at most $O(1)$ solutions.
\end{reptheorem}

\begin{proof}
We carry out the preprocessing 
in two stages. Each stage considers its input instances
$(\vec{a}, t)$ one at a time and produces one or more instances 
$(\vec{a}',t')$ for the next stage, the output of the second stage 
being the output of the procedure. 

The first stage takes as input the instance $(\vec{a}, t)$ given
as input to the algorithm.
Without loss of generality we may assume that 
$(\vec{a},t)$ satisfies $a_i\leq t$ for all $i=1,2,\ldots,n$.
Indeed, we may simply remove all elements $i$ with $a_i>t$. 
Hence $0\leq \sum_{i=1}^n a_i x_i\leq nt$ for all $\vec{x}\in\{0,1\}^n$.
A further immediate observation is that we may assume that 
$\log nt\leq 2^n$. Indeed, otherwise we can do an exhaustive search 
over all the $2^n$ subsets of the input integers in polynomial 
time in the input size (and then output a trivial YES or NO instance 
based on the outcome without proceeding to the second stage).
Next, select a uniform random prime $P$ with, say, $3n+1$ bits. 
For each $k=0,1,2,\ldots,n-1$, form one instance $(\vec{a}', t')$
by setting $t'=t\bmod P+kP$ and
$a_i'=a_i\bmod P$ for $i=1,2,\ldots,n$. Observe that every solution 
of $(\vec{a},t)$ is a solution of $(\vec{a}',t')$ for at least
one value of $k$. We claim that with high probability
each of the $n$ instances $(\vec{a}',t')$ has no other
solutions beyond the solutions of $(\vec{a},t)$.

Consider an arbitrary vector
$\vec{x}\in\{0,1\}^n$ that is not a solution of $(\vec{a},t)$
but is a solution of $(\vec{a}',t')$. This happens 
only if $P$ divides $Z=|t-\sum_{i=1}^n a_ix_i|\neq 0$.
Let us analyze the probability for the event that $P$ divides $Z$.
Since $Z\leq nt$ has at most $2^n$ bits (recall that $\log nt\leq 2^n$), 
there can be at most $2^n/(3n)$ primes with $3n+1$ bits that divide $Z$.
By the Prime Number Theorem we know
that there are $\Omega(2^{3n+1}/n)$ primes with $3n+1$ bits.
Since $P$ is a uniform random prime with $3n+1$ bits, 
we have that $P$ divides $Z$ with probability $O(2^{-2n}n^2)$.
By linearity of expectation, the expected number of vectors 
$\vec{x}\in\{0,1\}^n$ that are not solutions of $(\vec{a},t)$ 
but are solutions of $(\vec{a}',t')$ is thus $O(2^{-n}n^2)$.
By an application of Markov's inequality and the union bound,
with probability $1-o(1)$
each of the $n$ instances $(\vec{a}',t')$ has no other
solutions beyond the solutions of $(\vec{a},t)$.
By construction, $\log t'=O(n)$. This completes the first stage. 

The second stage controls the number of solutions by a standard
isolation technique.
Consider an instance $(\vec{a},t)$ input to the second stage.
Assume that the set of all solutions $S\subseteq\{0,1\}^n$ of $(\vec{a},t)$ 
is nonempty and guess that it has size in the range 
$2^s\leq |S|\leq 2^{s+1}$ for $s=0,1,\ldots,n-1$. 
(That is, we try out all values and at least one will be the correct guess.)
Select (arbitrarily) a prime $P$ in the interval $2^s\leq P\leq 2^{s+1}$.
Select $r_1,r_2,\ldots,r_n$ and $u$ independently and uniformly at
random from $\{0,1,\ldots,P-1\}$. 

For any fixed $\vec{x}\in S$, we have that 
\begin{equation}
\label{eq:isolation-mod-p}
\sum_{i=1}^n r_ix_i\equiv u\pmod P
\end{equation}
holds with probability $1/P$ over the random choices 
of $r_1,r_2,\ldots,r_n,u$. Similarly, any distinct 
$\vec{x},\vec{x'}\in S$ both satisfy \eqref{eq:isolation-mod-p} 
with probability $1/P^2$.

Fix a correct guess of $s$, so that $1 \le |S|/P \le 2$, and let the
random variable $S_P$ be the number of solutions in $S$ that also
satisfy \eqref{eq:isolation-mod-p}.  Letting $\lambda = |S|/P$ we then have
\begin{align*}
\E[S_P] &= \lambda & \text{and} && \E[S_P^2] = \E[S_P] + \frac{|S|(|S|-1)}{P^2} < \lambda + \lambda^2,
\end{align*}
so the first and second moment methods give
\begin{align*}
  \Pr[S_P > 10] & < \frac{\E[S_P]}{10} = \frac{\lambda}{10} < 1/5\quad \textrm{and} \\
  \Pr[S_P > 0] & > \frac{\E[S_P]^2}{\E[S_P^2]} > \frac{1}{1+\lambda} > 1/2\,.
\end{align*}
By a union bound, we have that for this correct
guess of $s$ at least $1$ and at most $10$ of the solutions in $S$
satisfy \eqref{eq:isolation-mod-p} with probability at least $1/4$.

Let $\vec{x}\in S$ satisfy \eqref{eq:isolation-mod-p}.  Then, there
exists a $k=0,1,\ldots,n-1$ such that $\sum_{i=1}^n r_ix_i=u+Pk$.
(Note that this is equality over the integers, not a modular
congruence!)  Again we can guess this value $k$ by iterating over all $n$
possibilities.  
Put $a_i'=a_i+(nt+1)r_i$ for $i=1,2,\ldots,n$ 
and $t'=t+(nt+1)(u+Pk)$. 

Now observe that if $S$ is empty, then none of
the $n^2$ instances $(\vec{a}',t')$ has solutions with probability 1. 
Conversely, if $S$ is nonempty, then at least one of the 
instances $(\vec{a}',t')$ has at least 1 and at most 10 solutions with 
probability at least 1/4. By construction, $\log t'=O(n)$.
Since the first stage gives $n$ outputs, the second stage
gives $n^3$ outputs in total.
\end{proof}

\section{Parallelization}
\label{sec:parallel}

In this section we prove Theorem~\ref{thm:parallel dissect}, restated
here for convenience.

\begin{reptheorem}{thm:parallel dissect}
  The algorithm of Theorem~\ref{thm:main} can be implemented to run in
  $O^*(2^{\tau(\sigma) n} / P)$ parallel time on $P$ processors each
  using $O^*(2^{\sigma n})$ space, provided $P \le
  2^{(2\tau(\sigma)-1)n}$.
\end{reptheorem}

\begin{proof}[Proof]
  We divide the $P$ processors evenly among the roughly $2^{\beta n}$
  choices of $s'$ in line~\ref{step:s-loop}. If $P \le 2^{\beta n}$,
  then this trivially gives full parallelization. Otherwise, fix a
  choice of $s'$. We have $P' \approx P/2^{\beta n}$ processors
  available to solve the instance restricted to this value of $s'$.

  We now let each of the $P'$ available processors solve the left 
  recursive call on line~\ref{step:firstrec} in full, independently of
  each other.  Only in the right  recursive call on
  line~\ref{step:secondrec} do we split up the task and use the $P'$
  processors to get a factor $P'$ speedup, provided that $P'$ is not
  too large (cf. the theorem statement).

  Let us write $\sigma_l$ and $n_l$ (resp.~$\sigma_r$ and $n_r$)
  for the values of $\sigma$ and $n$ on the left (resp.~right) 
  recursive branch.
  The left branch takes time $O^*(2^{\tau(\sigma_l) n_l})$. By
  an inductive argument, if $P' \le 2^{(2\tau(\sigma_r)-1)n_r}$, 
  then the right  branch takes time $O^*(2^{\tau(\sigma_r)n_r}/P')$. 
  Indeed, to set up the induction, observe that in the base case 
  when $\sigma \ge 1/4$, there is nothing to prove, since the bound 
  on $P$ is then simply $1$. The overall time taken is within 
  a constant of the maximum of these because the recursion depth is $O(1)$.

  Thus to complete the proof it suffices to establish the inequalities
  \begin{align}
    \label{eqn:parallel ineq 1}
    \max\big\{2^{\tau(\sigma_l)n_l},\, 2^{\tau(\sigma_r) n_r}/P'\big\} &\le 2^{\tau(\sigma) n}/P\,,\\
    \label{eqn:parallel ineq 2}
    P' &\le 2^{(2\tau(\sigma_r)-1) n_r}\,.
  \end{align}

  Let us start with \eqref{eqn:parallel ineq 1}. 
  For the left branch, we have $n_l = \alpha n
  = (1-\tau(\sigma))n$.  Using the assumption that $P \le
  2^{(2\tau(\sigma)-1)n}$ and the trivial bound $\tau(\sigma_l) \le 1$,
  we see that $2^{\tau(\sigma_l) n_l} \le 2^{\tau(\sigma) n}/P$ as
  desired. For the right branch, we have
  \begin{align*}
    n_r &= (1-\alpha)n = \tau(\sigma) n\,,\\
    \tau(\sigma_r) &= \tau(\sigma/\tau(\sigma)) =
    \frac{2\tau(\sigma)-1+\sigma}{\tau(\sigma)}\,,
  \end{align*}
  where the last step uses Proposition~\ref{prop:tau recurrence}.  
  Thus,
  \[
    \tau(\sigma_r) n_r = (2\tau(\sigma)-1+\sigma)n\,, 
  \]
and hence,
  \begin{equation*}
    2^{\tau(\sigma_r) n_r}/P' = 2^{(2\tau(\sigma)-1+\sigma)n}/(P/2^{(1-\tau(\sigma)-\sigma)n}) = 2^{\tau(\sigma) n}/P\,.
  \end{equation*}
  
  It remains to establish \eqref{eqn:parallel ineq 2}. 
  Because $P \le 2^{(2\tau(\sigma)-1)n}$, it suffices to show that
  \begin{equation*}
    2^{(2\tau(\sigma)-1)n}/2^{(1-\tau(\sigma)-\sigma)n} \le 2^{(2\tau(\sigma_r)-1)n_r} = 2^{(4\tau(\sigma)-2+2\sigma - \tau(\sigma))n}\,.
  \end{equation*}
  Canceling exponents on the left and on the right, everything cancels 
  except for one of the two $\sigma n$'s on the right.
\end{proof}

\bibliographystyle{plain}
\bibliography{paper}

\begin{thebibliography}{10}

\bibitem{BeckerCoronJoux11}
Anja Becker, Jean-S{\'e}bastien Coron, and Antoine Joux.
\newblock Improved generic algorithms for hard knapsacks.
\newblock In Kenneth~G. Paterson, editor, {\em EUROCRYPT}, volume 6632 of {\em
  Lecture Notes in Computer Science}, pages 364--385. Springer, 2011.

\bibitem{BjorklundHusfeldtKaskiKoivisto08}
Andreas Bj{\"o}rklund, Thore Husfeldt, Petteri Kaski, and Mikko Koivisto.
\newblock Computing the {T}utte polynomial in vertex-exponential time.
\newblock In {\em FOCS}, pages 677--686. IEEE Computer Society, 2008.

\bibitem{DinurDunkelmanKellerShamir12}
Itai Dinur, Orr Dunkelman, Nathan Keller, and Adi Shamir.
\newblock Efficient dissection of composite problems, with applications to
  cryptanalysis, knapsacks, and combinatorial search problems.
\newblock In Reihaneh Safavi-Naini and Ran Canetti, editors, {\em CRYPTO},
  volume 7417 of {\em Lecture Notes in Computer Science}, pages 719--740.
  Springer, 2012.

\bibitem{HorowitzSahni74}
Ellis Horowitz and Sartaj Sahni.
\newblock Computing partitions with applications to the knapsack problem.
\newblock {\em J. ACM}, 21(2):277--292, April 1974.

\bibitem{Howgrave-GrahamJoux10}
Nick Howgrave-Graham and Antoine Joux.
\newblock New generic algorithms for hard knapsacks.
\newblock In Henri Gilbert, editor, {\em EUROCRYPT}, volume 6110 of {\em
  Lecture Notes in Computer Science}, pages 235--256. Springer, 2010.

\bibitem{karp72}
Richard~M. Karp.
\newblock Reducibility among combinatorial problems.
\newblock In Raymond~E. Miller and James~W. Thatcher, editors, {\em Complexity
  of Computer Computations}, The IBM Research Symposia Series, pages 85--103.
  Plenum Press, New York, 1972.

\bibitem{KaskiKoivistoNederlof12}
Petteri Kaski, Mikko Koivisto, and Jesper Nederlof.
\newblock Homomorphic hashing for sparse coefficient extraction.
\newblock In Dimitrios~M. Thilikos and Gerhard~J. Woeginger, editors, {\em
  IPEC}, volume 7535 of {\em Lecture Notes in Computer Science}, pages
  147--158. Springer, 2012.

\bibitem{KoivistoParviainen10}
Mikko Koivisto and Pekka Parviainen.
\newblock A space-time tradeoff for permutation problems.
\newblock In Moses Charikar, editor, {\em SODA}, pages 484--492. SIAM, 2010.

\bibitem{LokshtanovNederlof10}
Daniel Lokshtanov and Jesper Nederlof.
\newblock Saving space by algebraization.
\newblock In Leonard~J. Schulman, editor, {\em STOC}, pages 321--330. ACM,
  2010.

\bibitem{SchroeppelShamir81}
Richard Schroeppel and Adi Shamir.
\newblock A ${T}={O}(2^{n/2})$, ${S}={O}(2^{n/4}$) algorithm for certain
  {NP}-complete problems.
\newblock {\em SIAM J. Comput.}, 10(3):456--464, 1981.

\bibitem{Woeginger08}
Gerhard~J. Woeginger.
\newblock Open problems around exact algorithms.
\newblock {\em Discrete Applied Mathematics}, 156(3):397--405, 2008.

\end{thebibliography}

\end{document}